%% file: manuscript.tex
\documentclass[footinbib,prl,twocolumn,amsmath,amssymb,aps]{revtex4-1}
\usepackage{graphicx}
\usepackage{hyperref}

\hypersetup{
     colorlinks   = true,
     linkcolor    = blue,
     citecolor    = blue
}

\usepackage{amsthm}
\usepackage{physics}

\newtheorem{proposition}{Proposition}
\begin{document}

\title{Topological order and criticality in (2+1)D monitored random quantum circuits}
\author{ {$\rm Ali \ Lavasani^{1,2}$}, {$\rm Yahya \ Alavirad^{1,2,3}$, $\rm Maissam \ Barkeshli^{1,2}$}}
\affiliation{$^1$Condensed Matter Theory Center, University of Maryland, College Park, Maryland 20742, USA}
\affiliation{$^2$Joint Quantum Institute, University of Maryland, College Park, Maryland 20742, USA}
\affiliation{$^3$Department of Physics, University of California at San Diego, La Jolla, CA 92093, USA}

\begin{abstract}
It has recently been discovered that random quantum circuits provide an avenue to realize rich entanglement phase diagrams, which are hidden to standard expectation values of operators. Here we study (2+1)D random circuits with random Clifford unitary gates and measurements designed to stabilize trivial area law and topologically ordered phases. With competing single qubit Pauli-Z and toric code stabilizer measurements, in addition to random Clifford unitaries, we find a phase diagram involving a tricritical point that maps to (2+1)D percolation, a possibly stable critical phase, topologically ordered, trivial, and volume law phases, and lines of critical points separating them. With Pauli-Y single qubit measurements instead, we find an anisotropic self-dual tricritical point, with dynamical exponent $z \approx 1.46$, exhibiting logarithmic violation of the area law and an anomalous exponent for the topological entanglement entropy, which thus appears distinct from any known percolation fixed point. The phase diagram also hosts a measurement-induced volume law entangled phase in the absence of unitary dynamics.
\end{abstract}

\maketitle

\emph{Introduction} - In the past few years, it has been realized that the interplay between measurements and unitary dynamics can give rise to rich physics in the dynamics of quantum entanglement~
\cite{skinner2019measurement,li2018quantum,chan19,PhysRevB.100.134306,gullans2019dynamical,PhysRevLett.125.070606,ludwing19tensor,altman19error,Szyniszewski19weak,tang2020dmrg,ludwig20replica,deluca19fermion,Vasseur20tensor,altman19phasetheory,qi20blackhole,jed20numeric,vicari20ising,iaconis2020,fuji2020,lang20,lavasani20,sang2020measurement,xiao20emerge,nahum20,lunt20,szyniszewski2020universality,gullans2020quantum,nahum19perc,alberton2020trajectory,turkeshi2020measurement,fidkowski2020dynamical,vijay2020measurementdriven,fan2020self,ippoliti2020entanglement,li2020statistical,van2020entanglement,sh2020classical}.
Originally, it was shown that when $(1+1)$D random unitary dynamics are intercepted by local measurements at a rate $p$, the system can undergo a phase transition from a volume law entangled phase at $p<p_c$ to an area law entangled phase at $p>p_c$~\cite{skinner2019measurement,li2018quantum,chan19}. Importantly, these phase transitions are entirely hidden from simple expectation values of operators but are manifest in \textit{quantum-trajectory averaged} dynamics of entanglement measures, like the entanglement entropy (EE) ~\cite{skinner2019measurement, lavasani20,li2018quantum,PhysRevB.100.134306}.

 \begin{figure}
   \includegraphics[width=0.95\columnwidth]{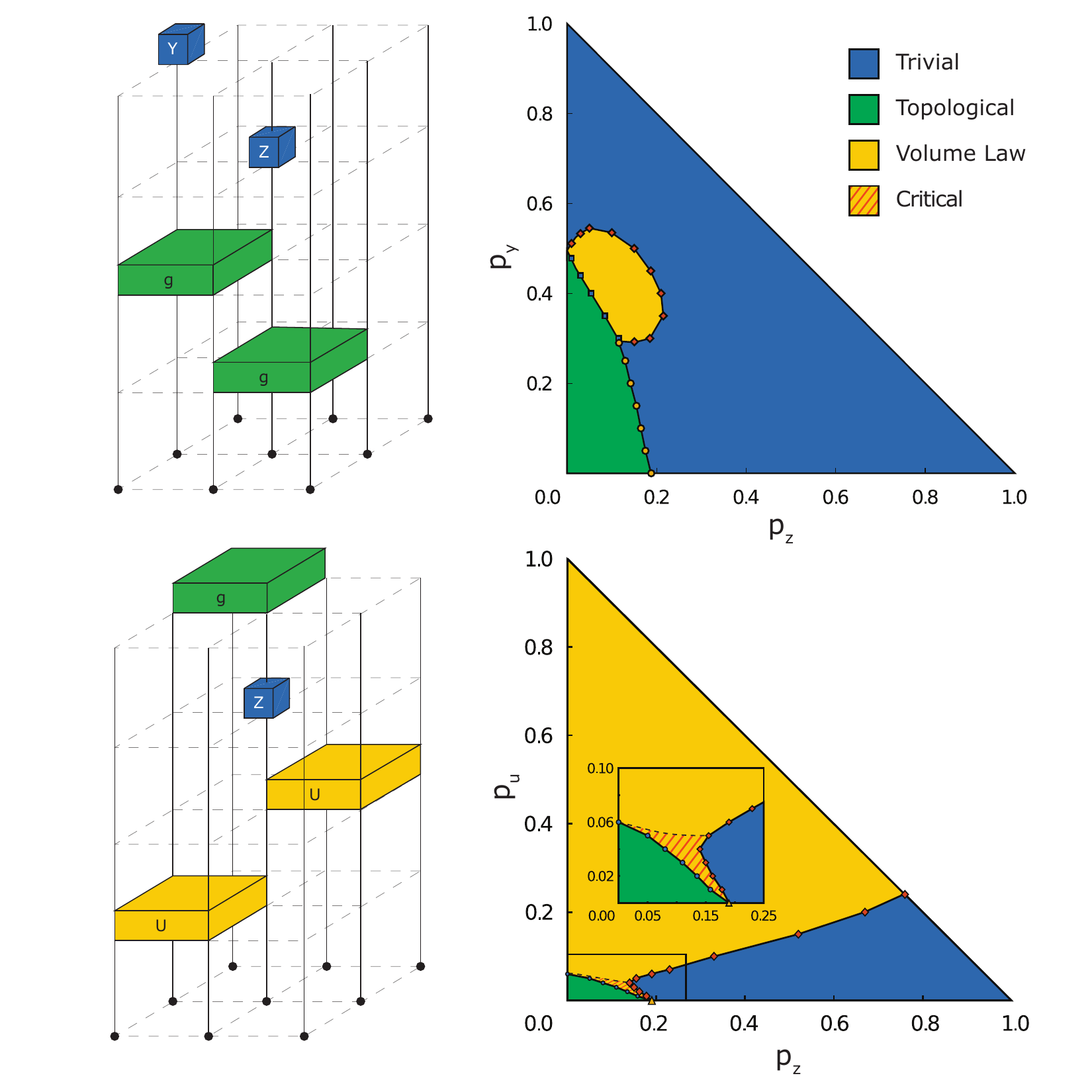}
   \caption{(a) A typical measurement-only random circuit. (b) Phase diagram of (2+1)D measurement-only random circuits. (c) A typical  hybrid random circuit. (d) Phase diagram of (2+1)D hybrid random circuits. (e) Entanglement dynamics at the $p_y=0$ line of  measurement-only random circuits (as well as the $p_u=0$ line of hybrid random circuits) maps to a classical bond percolation problem on a cubic lattice.}
   \label{fig_circuits}
 \end{figure}
Recently, it was shown that competing measurements can give rise to entanglement transitions even in the absence of unitary dynamics~\cite{nahum19perc,lavasani20,sang2020measurement,ippoliti2020entanglement}. Furthermore, it was discovered that distinct $(1+1)$D area law phases can remain well-defined in the context of random quantum circuits\cite{lavasani20,sang2020measurement}.

In this work, we consider a class of $(2+1)$D random quantum circuits that extrapolate between (1) a topologically ordered phase, characterized by non-zero topological entanglement entropy (TEE)~\cite{PhysRevLett.96.110404,levin06topo} and realized by measuring the $\mathbb{Z}_2$ toric code stabilizers~\cite{KITAEV20032}, (2) a volume law entangled phase realized by random Clifford unitaries and, (3) the trivial, area law phase realized by single-site measurements. As for the single-site measurements, we study both Pauli-Z and Pauli-Y measurements. This generalizes the work of Ref.~\cite{lavasani20} to $(2+1)$D where symmetry restrictions are not necessary. Similar to Ref.~\cite{lavasani20}, at each step of the circuit, an element corresponding to one of the three phases is selected at random with probability $p_g$, $p_u$, $p_s$ respectively (subject to the condition $p_g+p_u+p_s=1$).

Two typical arrangements of our circuits together with numerically calculated phase diagrams are shown in Fig.~\ref{fig_circuits}. We find stable topological, trivial area law, volume law (even without unitary dynamics), and critical phases. Notably, we also find evidence of several qualitatively distinct multi-critical points.

In the absence of unitary dynamics and in the case where single qubits are only measured in the Pauli-$Z$ basis, we find an exact analytical mapping between a $3$D classical bond percolation problem and the dynamics of entanglement.
 We show that EE of rectangular regions are related to the number of clusters shared between that region and the rest of the system in percolation picture.

On the other hand, when single qubits are only measured in the $Y$ direction,
we show that an Ising duality restricts the phase diagram. We find a novel tricritical point at the self-dual point $(p_z,p_y)=(0,0.5)$, separating topological, trivial area law, and volume law phases, in which the critical behavior of the system is qualitatively different from the rest of the phase diagram. Intriguingly, the circuit has non-trivial subsystem symmetries at this point.

Extensive numerical study of the phase diagram shows that \textit{away} from the self-dual point discussed above: (1) The critical points are area law entangled, similar to usual $(2+1)$D scale invariant field theories. The sub-leading correction to this area-law scaling also agrees with results found in a variety of $(2+1)$D scale invariant field theories~\cite{chen2015scaling,PhysRevB.85.165121,stephan2013entanglement,kallin2014corner} (in contrast to the results of Ref.~\cite{turkeshi2020measurement}). (2) We find a correlation length exponent $\nu=0.8(1)$ and a dynamical critical exponent $z=1$, which are set by the classical $3$D bond percolation theory. Within margin of error, these exponents stay constant along phase boundaries.

 However, the critical dynamics at the self-dual $p_z=0$, $p_y=0.5$ point is entirely distinct and characterized by: (1) Logarithmic corrections to the area law scaling of EE reminiscent of Fermi liquids. (2) Non-percolation correlation length exponent $\nu=0.47(8)$ and a dynamical critical exponent $z=1.46(6)$. (3) A non-zero anomalous $\gamma=1.0(2)$ exponent for the TEE (see Eq. \eqref{eq_scaling}).

 \emph{Models} - We consider $N=L^2$ qubits laid on the vertices of a two dimensional periodic square lattice of linear length $L_x=L_y=L$. Three different sets of gates are considered where each gate set, when applied exclusively, drives the system into one of distinct phases discussed above.

 For the topological phase, we consider measurements corresponding to toric code stabilizers,
 \begin{equation}\label{eq_gij}
   g_{i,j}=
   \begin{cases}
       X_{i,j}~X_{i+1,j}~X_{i,j+1}~X_{i+1,j+1} & i+j \text{ is even} \\
       Z_{i,j}~Z_{i+1,j}~Z_{i,j+1}~Z_{i+1,j+1} & i+j \text{ is odd}
    \end{cases},
 \end{equation}
 where $(i,j)$ denotes the coordinates and $X_{i,j}$ and $Z_{i,j}$ are the Pauli operators acting on the corresponding qubit. We denote the set of all $g_{i,j}$ operators as $\mathcal{M}_{g}$.

 For the trivial phase, we can pick any set of single qubit measurements. We use $\mathcal{M}_P$ to denote the set of single qubit Pauli-$P$ operators ($P$ could be either X,Y or Z).
For the volume law phase, we use the set $\mathcal{C}_4$ consisting of four qubit Clifford unitaries $U_{i,j}$, acting on neighboring qubits located at $(i,j)$, $(i+1,j)$, $(i,j+1)$ and $(i+1,j+1)$.

 We study two types of random circuits. First, we consider \emph{measurement-only random circuits} comprised of only measurements. More specifically, we start with the product state $\ket{0}^{\otimes N}$ and at each updating step, we measure an operator which is chosen uniformly at random from either $\mathcal{M}_Z$ with probability $p_z$, $\mathcal{M}_Y$ with probability $p_y$ or $\mathcal{M}_g$ with probability $p_g=1-p_z-p_y$. Each time step is defined as $N$ consecutive updating steps. A typical example of such a circuit is shown in Fig. \ref{fig_circuits}a.

We also consider \emph{hybrid random circuits}, which are comprised of unitary gates as well as measurements. We start with $\ket{0}^{\otimes N}$ and at each updating step we either apply a gate chosen uniformly at random from $\mathcal{C}_4$ with probability $p_u$ or measure an operator chosen uniformly at random from $\mathcal{M}_Z$ or $\mathcal{M}_g$ with probabilities $p_z$ and $p_g=1-p_u-p_z$ respectively.

 \emph{Order Parameters} - One can use TEE ~\cite{PhysRevLett.96.110404,levin06topo} denoted by $S_\text{topo}$ to distinguish phases.   $S_\text{topo}$ equals $1$ for the eigenstates of the toric code Hamiltonian while it is $0$ for quantum states in the trivial phase. As for the volume law phase, the contribution which is proportional to the size of each region cancels out and one may expect $S_\text{topo}$ to vanish in this phase as well. However, the $(1+1)$D results\cite{PhysRevB.100.134306,li2020statistical}
  suggest that the EE of a region has sub-extensive contributions ~\cite{fan2020self,li2020statistical} in the volume law phase, which results in a system-size dependent value for $S_\text{topo}$. Our numerical results support this scenario.

We also utilize the ancilla order parameter introduced in Refs.\cite{gullans2019dynamical,PhysRevLett.125.070606} which captures the transition in purification dynamics. It is defined using $N_a$ ancilla qubits in addition to the system qubits, as follows. First a random local Clifford unitary circuit of depth $O(N)$ is applied to the entire set (system + ancilla) of qubits, which results in a maximally entangled stabilizer state of all qubits. Next, the system qubits are evolved under the random quantum circuit of interest for $T$ time-steps and then the EE of the set of all ancilla qubits, denoted by $S_a(T)$, is measured. In the $T\to \infty$ limit, the ancilla system will be entirely disentangled from the system. However, this purification dynamics happens with a rate which depends on the phase of the system. In the trivial area law phase, the ancilla qubits will be disentangled in constant time, independent of the system size. In the topological phase, although the bulk disentangles in constant time, the logical qubits remain entangled with the ancilla system until a time exponentially large in system size. In the volume law phase, the bulk remains entangled with the ancilla qubits up to exponentially large time-steps. Therefore for large enough system sizes and at $T=\mathcal{O}(L)$, $S_a(T)$ will be $0$, $N_L$ and $N_a$ in trivial, topological and volume law phases respectively, where $N_L$ denotes the number of logical qubits in the topological phase ($N_L=2$ for the torus topology). We assumed that $N_L \ll N_a \ll L^2$. We use $N_a=10$ ancillas throughout this work.

  We note that while in our setting the purification transition occurs concurrently with TEE phase transition, they are not exactly the same\cite{gullans2019dynamical}. One can, for instance, repeat the study here on a triangulation of a 2-sphere, for which $N_L = 0$, so purification protocols cannot distinguish the trivial and topological phases, while TEE can.

\emph{Results} -
 We start by studying the phase diagram of the measurement-only circuits.
 First, we focus on the $p_y = 0$ line. Notably, as shown in the Supplemental Material, there is an exact mapping which maps the entanglement dynamics at this line of the phase diagram to a classical bond percolation problem on a $3$D cubic lattice. Fig. \ref{fig_yz_main}a and b show the TEE and the ancilla order parameter as a function of $p_z$. As can be seen from the plots, there exists a stable topological phase extending up to $p_c \approx 0.2$, at which point a continuous phase transition takes the system to the trivial phase.

 \begin{figure*}
   \includegraphics[width=\textwidth]{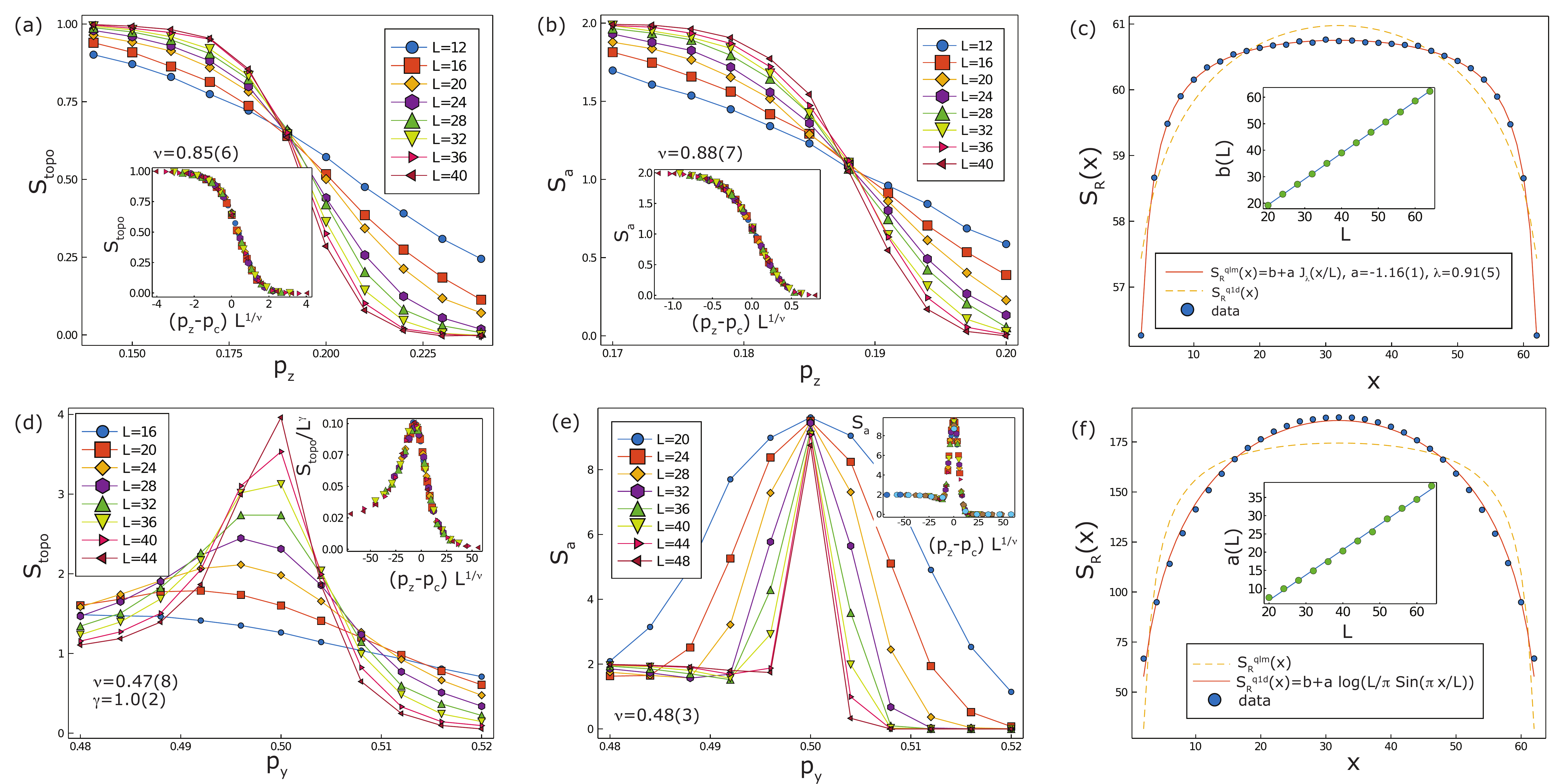}
   \caption{
   Phase transitions across the $p_y=0$ (top row) and $p_z=0$ (bottom row) lines of the phase diagram for the measurement-only circuit: (a) $S_\text{topo}$ and (b) $S_a$ measured at $t=4L$ versus $p_z$ for fixed $p_y=0$. Insets show the corresponding data collapse. (c) $S_R(x)$ for system size $L=64$ at the percolation critical point $(p_z,p_y)=(0.188 ,0)$, with the best fit of scaling functions $S^\text{qlm}(x)$ (solid line) and $S^\text{q1D}$ (dashed line). The inset is the best fit value of the $b$ parameter in Eq.\eqref{eq_s_qlm} as a function of $L$. (d) $S_\text{topo}$ and (e) $S_a$ measured at $t=0.6~L^{1.46}$
    versus $p_y$ for fixed $p_z=0$. Insets show the corresponding data collapse. (f) $S_R(x)$ for system size $L=64$ at the self-dual critical point $(p_z,p_y)=(0,0.5)$, with the best fit of scaling functions $S^\text{q1D}$ (solid line) and $S^\text{qlm}(x)$ (dashed line). The inset shows the linear dependence of the best fit value of the $a$ parameter in Eq.\eqref{eq_sr_1d}.
   }
   \label{fig_yz_main}
 \end{figure*}

 On general grounds, we may assume the following scaling forms governing the order parameters near the phase transition
 \begin{align}\label{eq_scaling}
   &S_\text{topo}(p;L)=L^{\gamma} F((p-p_c) L^{1/\nu}), \\
  &S_a(p,t,L)=G((p-p_c) L^{1/\nu},t/L^z), \label{sq_scaling_Sa}
 \end{align}
 where $F(x)$ and $G(x)$ are arbitrary functions and $\nu$ and $z$ are the correlation length critical exponent and dynamical critical exponent respectively. We find our data for the percolation critical point to be consistent with setting $\gamma$ to $0$. By collapsing $S_\text{topo}$ near the critical point for different system sizes, we find $p_c=0.188(2)$ and $\nu=0.85(6)$. Note that $\nu$ is consistent with the values obtained from numerical simulation of classical percolation in $3$D\cite{PhysRevE.87.052107}. By investigating the time dependence of the ancilla order parameter $S_a$ at $p=p_c$,
  we find it to be consistent with $z=1$ (see Supplemental Material for relevant plots). Collapsing $S_a$ at $t=O(L)$ then yields $\nu=0.88(7)$, in agreement with the value found via collapsing $S_\text{topo}$.

 Another quantity of interest is the scaling form of the EE with sub-system size at the critical point.
  We consider the cylindrical region $R$ with a smooth boundary, which has length $x$ in one direction and goes all the way around the torus in the other direction. Let $S_R(x)$ denote its EE. Note that the boundary length $|\partial R|$ is $2L$, independent of $x$. As is discussed in the Supplemental Material, in the percolation picture this quantity is related to the number of clusters with shared support on region $R$ and its complement.

  For a conventional CFT in $(2+1)$D, the non-universal leading area-law term scales with $|\partial R|=2L$. The sub-leading term for a cylindrical subregion is less well-understood  and several forms have been suggested, among which two are of particular interest. One is a quasi-$(1+1)$D scaling function, inspired by the exact form found in $(1+1)$D CFTs, which seems to decently capture EE scaling in certain (2+1)D gapless models\cite{PhysRevB.85.165121}
\begin{equation}\label{eq_sr_1d}
  S^{\text{q1d}}_R(x)=b+a~\log(\frac{L}{\pi}\sin(\frac{\pi~x}{L})),
\end{equation}
where $b$ contains the non-universal area law term. The other relevant scaling form was originally derived for the quantum Lifshitz model\cite{stephan2013entanglement} but was found to describe the EE scaling in various other (2+1)$D$ gapless models as well, including some (2+1)D CFTs \cite{chen2015scaling}.
\begin{align}\label{eq_s_qlm}
  &S^{\text{qlm}}_R(x)=b+a J_\lambda(x/L)\\
  &J_\lambda(u)=\log\qty(\frac{\theta_3(i \lambda u) \theta_3(i \lambda (1-u))}{\eta(2iu)\eta(2i(1-u))}),
\end{align}
where $\theta_3(z)$ and $\eta(z)$ are the Jacobi theta function and the Dedekind eta function respectively (see Supplemental Material for definitions). $b$ contains the non-universal area-law contribution and $\lambda$ is a model parameter, which we will use to find the best fit.

Fig. \ref{fig_yz_main}c shows $S_R(x)$ for system size $L=64$ at $p_c$ alongside the best fit of the scaling functions. As can be seen from the graph, $S^\text{qlm}_R(x)$ results in a good fit (solid line), while $S_R^\text{q1D}$ cannot capture the scaling form. Moreover, we find that the best fit values of $a=-1.16(1)$ and $\lambda=0.91(5)$ for $S^{qlm}_R$ remain constant for different system sizes within the margin of error (see Supplemental Material). As is shown in the inset, the $b$ parameter scales linearly with system size, which shows that the leading term scales with $|\partial R|$.

We now turn our attention to the $p_z=0$ line. Here the circuit has a self-duality mapping $p_y \rightarrow 1 - p_y$. Note that along this line, the system has $2L$ subsystem symmetries generated by the product of $Y$ (or stabilizer) operators along horizontal or vertical loops, e.g. $\prod_j Y_{i,j}$. On a related note, there is a unitary transformation which maps the $g$ and $Y$ operators to the gauge operators of the $2D$ Bacon-Shor subsytem code\cite{bacon2006operator,shor1995scheme} on a square lattice (see the Supplemental Material for details).

 By examining the TEE $S_\text{topo}$(Fig. \ref{fig_yz_main}d), we find the topological phase to be extended up to the self-dual point $p_y=0.5$. However, at $p_y=0.5$, $S_\text{topo}$ grows with system size, which suggests a non-zero $\gamma$ exponent. Collapsing $S_\text{topo}$ data near the critical point yields $p_c=0.502(1)$, $\gamma=1.0(2)$ and $\nu=0.47(8)$, which shows that this critical point is distinct from the percolation fixed point.
Moreover, by looking at the time dependence of the ancilla order parameter at $p_y=0.5$, we find that, in contrast to the percolation critical point, the best fit to the scaling form in Eq. \ref{eq_scaling} corresponds to $z=1.46(8)$ (see Supplemental Material for relevant plots). Accordingly, by collapsing $S_a(p,t,L)$ data at $t=O(L^{1.46})$ (Fig. \ref{fig_yz_main}e), we find $\nu=0.48(3)$, in agreement with the result obtain from collapsing $S_\text{topo}$.

As for the cylindrical subregion EE $S_R(x)$ (Fig. \ref{fig_yz_main}f), we find that the quasi-1d scaling form $S^\text{q1d}(x)$ -- rather than $S^\text{qlm}(x)$ -- fits the data. However, as is shown in the inset, the $a$ parameter in Eq.\eqref{eq_s_qlm} is not constant, but has a linear dependence on system size $L$, demonstrating that the leading term scales as $L \log L$ rather than $L$ as is expected in an area law state. The origin of the $L \log L$ violation is unclear; it may be related to the existence of subsystem symmetries along the $p_z = 0$ axis, which translates to the stabilizers of the Bacon-Shor code under the aforementioned duality map.

 The rest of the phase diagram can be determined analogously (Fig. \ref{fig_circuits}a). We find that the percolation critical point is part of a critical line that persists up to some finite non-zero value of $p_y$, while the self-dual critical point at $(p_z,p_y)=(0,0.5)$ splits into two critical lines with an intermediate volume law entangled phase in between, making it tricritical. Interestingly, the numerical data for all critical points that we considered, other than $(p_z,p_y)=(0,0.5)$, are consistent with $z=1$ and $\gamma=0$, with $\nu$ remaining close to $0.8$, similar to the percolation critical point. Their EE scaling is given by $S^\text{qlm}(x)$ as well, with an area law scaling leading term. Remarkably, this makes the self-dual point special in this regard, as it is the only point in the phase diagram with $L \log L$ violation of area law, as well as quite different $\nu$ and $\gamma$ exponents.
  We also note that the extracted $a$ and $\lambda$ parameters in $S^\text{qlm}(x) $ change throughout the phase diagram. The relevant plots can be found in the Supplemental Material.

  lastly we present the numerical results for the hybrid random circuit which has unitary dynamics. The $p_u=0$ line of phase diagram is exactly the same as $p_z=0$ line of the measurement-only random circuit. Fig.~\ref{fig_pu001}a shows the ancilla order parameter along $p_u=0.01$, which signals the emergence of an intervening phase between topological and trivial phases, suggesting that the percolation critical point is actually a tricritical point in this phase diagram. In the intermediate phase, $S_a$ does not saturate to $N_a = 10$, as is expected to be the case in the volume law phase, but rather increases weakly with system size, showing indications that it may saturate at a finite value less than $10$. Indeed, for a point in the intermediate phase and for large systems, $S_a(t;L)$ seems to be a function of only $t/L$ (Fig. \ref{fig_pu001}b) which is a signature of a critical phase with $z=1$ (see Eq. \ref{sq_scaling_Sa}). Moreover, we find that in the intermediate phase, $S^\text{qlm}(x)$ fits the EE of a cylindrical subregion as well. These points suggest that the intermediate phase is a critical region. Nonetheless, we remark that the observed behavior could be just related to finite size effects and the proximity to the critical lines.

The critical region extends to the $p_u$ axis, ending at $p_u\approx 0.06$, which appears to be a tricritical point, although within the precision of this study, we cannot rule out the existence of a narrow critical region around $p_u = 0.06$. By collapsing the $S_a$ data along the $p_u$ axis, we find $p_{c}=0.059(1)$ with critical exponent $\nu=0.78(8)$.  On the right, the critical region ends on the boundary of the trivial phase and the volume law phase. The trivial phase itself ends at $p_u=0.238(2)$ along the $p_u+p_z=1$ line. We find  $\nu=0.80(7)$ at the corresponding phase transition  (see the Supplemental Material for relevant plots). The overall phase diagram of the hybrid circuit in $2$D is illustrated in Fig. \ref{fig_circuits}b.
 \begin{figure}
   \includegraphics[width=\columnwidth]{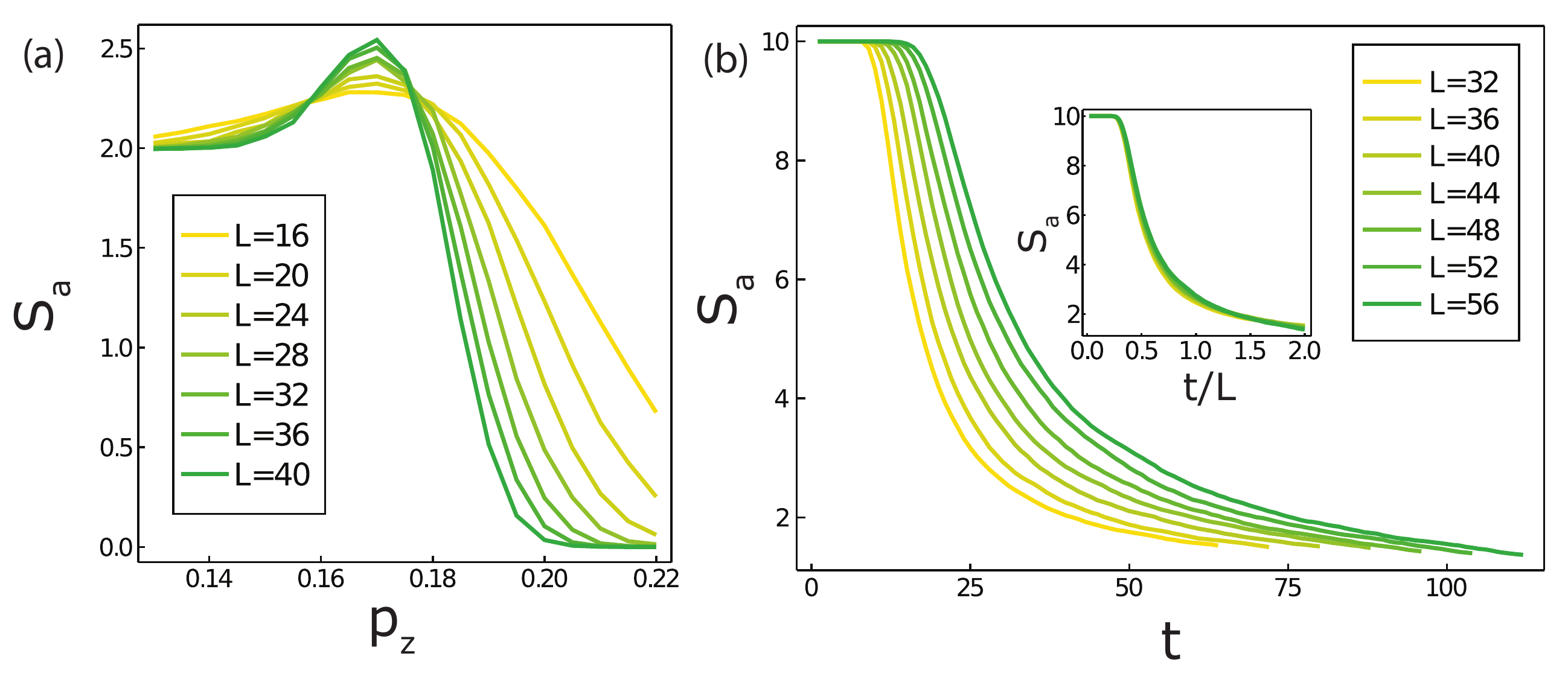}
   \caption{(a) The ancilla EE $S_a$ measured at $t=L$ as a function of $p_z$ for fixed $p_u=0.01$ in the hybrid random circuit. (b) $S_a$ as a function of time at $(p_z,p_u)=(0.17,0.01)$ in the hybrid random circuit. The inset is the same, plotted as a function of $t/L$.}
   \label{fig_pu001}
 \end{figure}

\emph{Acknowledgements} - We thank David Huse and Michael Gullans for discussions. The  authors  acknowledge  the  University  of  Maryland  supercomputing  resources (http://hpcc.umd.edu)  made  available  for  conducting  the  research reported in this paper. A.L and M.B are supported by NSF CAREER (DMR- 1753240) and Alfred P. Sloan Research Fellowship. Y.A is supported by Simons Collaboration on Ultra-Quantum Matter, grant 651440 from the Simons Foundation as well as University of California  Laboratory Fees Research Program, grant LFR-20-653926.

\bibliographystyle{apsrev4-1.bst}
\bibliography{refs.bib}

\clearpage
\input{SM.tex}

\end{document}

%% file: SM.tex
%
%
%
%
%
%
%
%
%

\appendix
\section{Supplemental Material}

\subsection{Percolation Mapping}\label{perc}

In this section we present the mapping between the random quantum circuit with only stabilizers and single qubit $Z$ measurements to a classical bond percolation problem on the 3D cubic lattice. The quantum circuit could be viewed as either the $p_y=0$ line of the projective random quantum circuit or the $p_u=0$ line of the hybrid random quantum circuit.

For simplicity, we consider the infinite plane geometry so we can ignore the non-trivial cycles of the torus as well as the subtleties arising near the boundary. Moreover, we work with the standard version of toric code where the qubits are on the edges and the $X$ and $Z$ stabilizers correspond to the star and plaquette operators respectively.

We take the initial state to be the eigenstate of all the star and plaquette operators. Note that since single qubit measurements are only in $Z$ direction, the system will remain to be an eigenstate of plaquette operators. Therefore, we ignore all plaquette operator measurements in what follows.

We place a square on each vertex which represents the corresponding star operator (Fig.\ref{fig_percolation}a). We start with all squares having a unique color. At each time step, the colors are updated as follows
\begin{itemize}
  \item Whenever a star operator  is measured, it acquires it own unique color (Fig.\ref{fig_percolation}b)
  \item Whenever a single qubit is measured, the two adjacent squares acquire the same color (Fig.\ref{fig_percolation}c)
\end{itemize}
As we will show, these rules represent how the stabilizer state of the system evolves. On the other hand, these set of rules arise naturally in a percolation picture. Consider a $3D$ cubic lattice, henceforth called the percolation lattice, and place a square on each vertex. The vertical direction corresponds to time. If a star operator is \textit{not} measured at time $t$, we connect the corresponding squares on the $(t-1)$'th and $t$'th layers. And if a qubit is measured at time $t$, we connect the corresponding adjacent squares on the $t$'th layer. At the end, we assign a unique color to each cluster of connected squares. The coloring that will arise is exactly what one would have found if the rules listed above were followed (see Fig.\ref{fig_percolation}).

The following proposition shows how the stabilizers specifying the state can be inferred from a given coloring configuration.
\begin{proposition}\label{prp_colorCoding}
  Let $A^j=\{s_i\}_{i=1}^n$ denote the set of squares corresponding to $j$'th color, in some arbitrary order. Up to a minus sign, the following operators stabilize the state:
  \begin{equation}\label{eq_stabs}
    \overline{\mathcal{X}}_j=\prod_{i=1}^n \mathcal{X}_{s_i}\qquad \text{and} \qquad Z_{s_i,s_{i+1}},
  \end{equation}
  where $\mathcal{X}_{s_i}$ is the star operator at square $s_i$ and $Z_{s_i,s_{i+1}}$ is the Pauli $Z$ string operator which starts on $s_i$ and ends on $s_{i+1}$.
\end{proposition}
Note that, because the quantum state is an eigenstate of every plaquette operator and because the space manifold has trivial topology, we don't need to specify the exact path that $Z_{s_i,s_{i+1}}$ string operator takes.
\begin{figure}
  \includegraphics[width=\columnwidth]{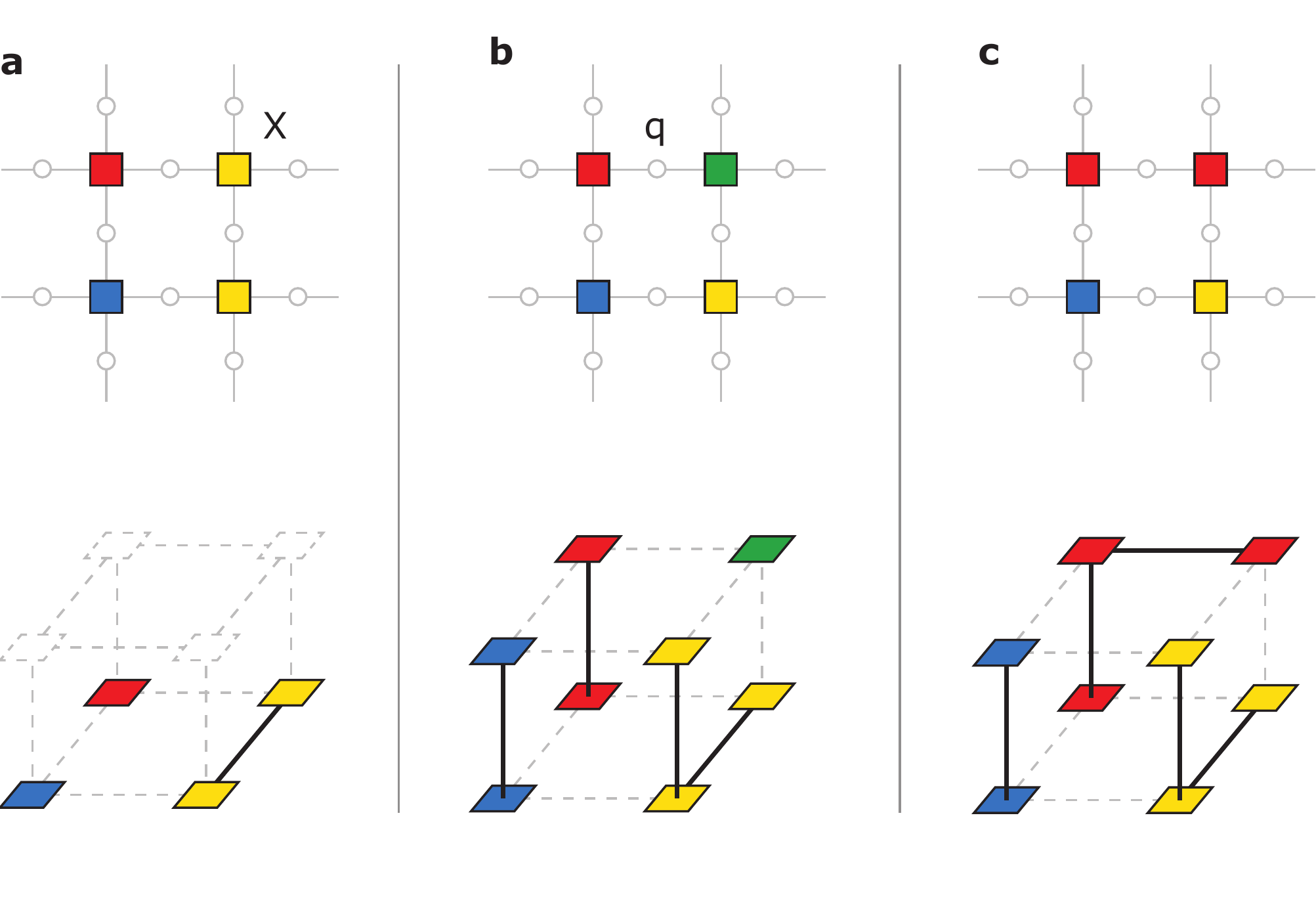}
  \caption{Diagrammatic representation of coloring rules outlined in the text arise naturally in the $3$D percolation picture. (a) The initial coloring of the squares (top) and the corresponding state in the percolation lattice (bottom).  (b) The state of the system after the star stabilizer marked by a $X$ in panel "a" is measured. (c) the state of the system after the qubit marked by $q$ in panel b is measured in $Z$ basis. }
  \label{fig_percolation}
\end{figure}
\begin{proof}[Proof of Proposition \ref{prp_colorCoding}]
We prove by induction. At $t=0$ it is clearly the case.

Now assume it is true at step $m$. Let's say we measure the star operator at some vertex $s_1$ at step $m+1$. If the corresponding square has a unique color already, nothing happens and the statement holds trivially afterwards. So let us consider the case where there are more than one squares with the same color as $s_1$, denoted by the set $A^1$. Based on the induction assumption, the quantum state is an eigenstate of $\prod_{s\in A^1} \mathcal{X}_s$ operator at step $m$. Therefore, after measuring $\mathcal{X}_{s_1}$,
 the system will be an eigenstate of $\mathcal{X}_{s_1}$ as well as $\mathcal{X}_{s_1}\times\prod_{s\in A^1}\mathcal{X}_{s}=\prod_{s\in A^1\setminus \{s_1\}} \mathcal{X}_{s}$. Moreover, any $Z_{i,j}$ string operator that doesn't end on $s_1$ commutes with $\mathcal{X}_{s_1}$ and as such, measurement of $\mathcal{X}_{s_1}$ has no bearing on the system being an eigenstate of it or not.
  As for the string operator stabilizers starting or ending on $s_1$, one of them will be replaced with $\mathcal{X}_{s_1}$ stabilizer, and the other with the product of the two which will be an string operator starting and ending on the set $A^1\setminus \{s_1\}$. Ergo the statement holds at step $m+1$ as well.

 Now consider the case in which, at step $m+1$, we measure a single qubit -- say  qubit $q$ which is between squares $s_1$ and $s_2$ -- in the $Z$ basis. If $s_1$ and $s_2$ have the same color, it results from the induction hypothesis that $Z_q$ already stabilizes the state and the statement trivially holds at step $m+1$. So let us consider the case where $s_1$ and $s_2$ have different colors,  such that $s_1 \in A^1$ and $s_2 \in A^2 $. Since $\prod_{s\in A^1 \cup A^2}\mathcal{X}_s$ stabilizes the state at step $m$ and commutes with $Z_q$, it stabilizes the state at step $m+1$ as well. Based on the induction assumption, the state is stabilized by any $Z_{i,j}$
 string that has both end points either on $A^1$ or $A^2$. Moreover, it is now stabilized by the string operator which connects $s_1\in A^1$ and $s_2 \in A^2$, namely $Z_q$. Therefore the state of the system at step $m+1$ is stabilized by any $Z_{i,j}$ string that has end points on $A^1 \cup A^2$. Hence the statement holds at step $m+1$ as well.
\end{proof}
Proposition \ref{prp_colorCoding} allows us to study the entanglement structure of the steady state, using the percolation picture. In particular, the entanglement entropy of a rectangular region $A$ is equal to the number of clusters which have shared support in $A$ and $A^c$, minus one.

\begin{figure}[t]
\includegraphics[width=0.8\columnwidth]{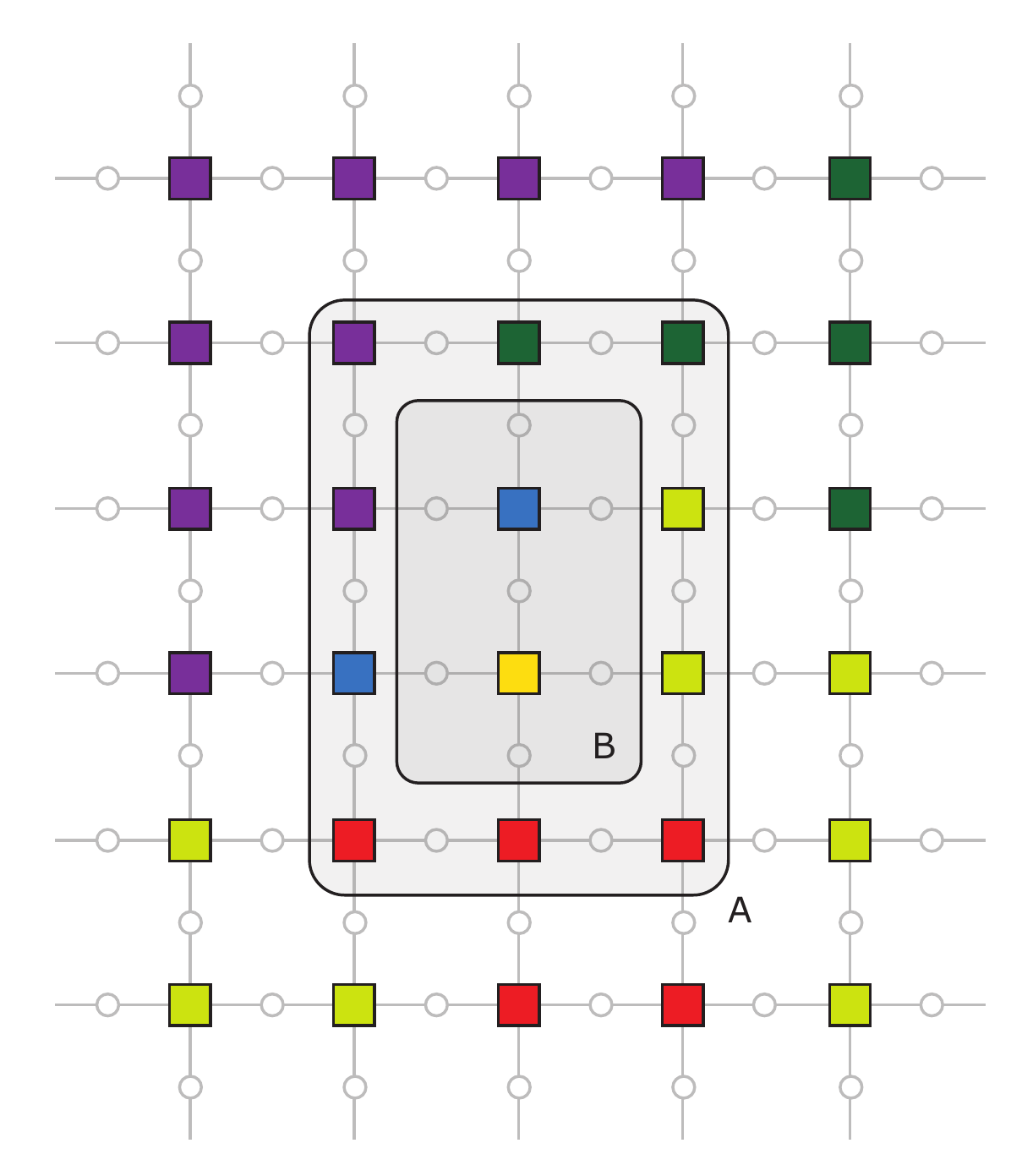}
\caption{An example of a rectangular region $A$ with sides $l_x=2$ and $l_y=3$. Its bulk is marked by $B$. In the example shown above, $C_A$ and $C^\text{ex}_B$ have $6$ and $1$ elements respectively.}
\label{fig_entropyofpercolation}
\end{figure}

\begin{proposition}\label{prp_perc_entropy}
  Let $A$ be a rectangular region with smooth boundary and let $B=A\setminus \partial A$ denote its bulk (see Fig.\ref{fig_entropyofpercolation}). For a given coloring configuration of the squares (star operators), let $C_A$ denote the set of colors that appear inside $A$. Also define $C^\text{ex}_B$ to be the colors that exclusively appear in the bulk. Then, for the stabilizer state which is associated with that coloring, the entanglement entropy of the subset of qubits inside $A$ can be expressed as,
  \begin{align}
  S_A=|C_A \setminus C^\text{ex}_{B}| -1.
  \end{align}
\end{proposition}

\begin{proof}
Let $A$ be a $l_x \times l_y$ rectangular region. By an abuse of notation, we use $A$ to denote the set of qubits which reside inside region $A$ as well. We prove this claim by direct calculation of the entanglement. For stabilizers states we have
  \begin{align}\label{brute}
  S_A=n_A-\dim G_A,
  \end{align}
where $n_A$ is the number of qubits in $A$ and $G_A$ is the subgroup of stabilizers which act trivially on the qubits in $A^c$. It is easy to see that $G_A$ can be generated by $Z$ plaquette stabilizers, $Z_{i,j}$ string stabilizers and $\overline{\mathcal{X}}_j$ stabilizers (see Proposition \ref{prp_colorCoding} for their definition) which are themselves contained in $A$. Let $n_{\square}$, $n_S$ and $n_{+}$ denote the number of independent $Z$ plaquette stabilizer, $Z_{i,j}$ string stabilizer and $\overline{\mathcal{X}}_j$
 stabilizers which are contained entirely within $A$, respectively. Then we have,
\begin{equation}
  S_A=n_A-n_\square-n_+-n_S
\end{equation}
Now we compute each quantity separately. finding $n_A$ and $n_\square$ is quite easy,
\begin{align}
  n_A&=2l_xl_y+l_x+l_y,\\
  n_\square&=l_x l_y.
\end{align}
Also we have $n_+=|C^\text{ex}_B|$, because any $\overline{\mathcal{X}}_j$ operator which has a star operator outside $B$ has nontrivial support in $A^c$. For a given color $c$, define $m_{c,A}$ to be the number of squares in $A$ with color $c$. Then we can write $n_S$ as,
\begin{equation}
  n_S=\sum_{c\in C_A} (m_{c,A}-1)=\sum_{c\in C_A} m_{c,A}-|C_A|.
\end{equation}
Now, the $\sum_{c\in C_A} m_{c,A}$ sum is just the total number of squares  inside $A$, which is $(l_x+1)(l_y+1)$. Thus we have
\begin{equation}
  n_S=l_x l_y+l_x+l_y+1-|C_A|.
\end{equation}
Putting everything together we get,
\begin{equation}
  S_A=|C_A|-|C^\text{ex}_B|-1=|C_A \setminus C^\text{ex}_B|-1,
\end{equation}
where in the last part we used the fact that $C^\text{ex}_B \subseteq C_A$.
\end{proof}

\subsection{Duality Mappings}\label{duality}
\begin{figure}
  \centering \includegraphics[width=0.7\columnwidth]{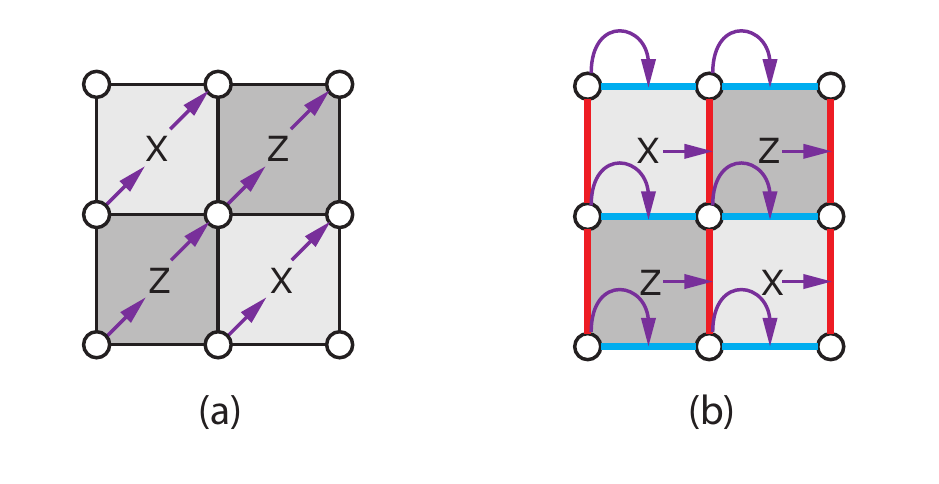}
  \caption{The diagrammatic illustration of the action of the duality map (a) $U$ defined by Eqs.\eqref{eq_duality_map} and (b) $\tilde{U}$ defined by Eqs.\eqref{eq_duality_map_bacon}. The white dots correspond to single qubit $Y$ operators and squares correspond to $g$ stabilizers. The vertical red lines and horizontal blue lines in panel b represent $XX$ and $ZZ$ operators respectively}
  \label{fig_duality_map}
\end{figure}
Here, we present the two dualities related to the $p_z=0$ line of the projective random circuits, one mapping the circuit at $p_y$ to the one at $1-p_y$ and the other mapping the $g$ and $Y$ operators to the gauge operators of the Bacon-Shor code \cite{xu2004strong, nussinov2015compass}. We start with the first one. Similar to previous section, we consider the infinite plane geometry.

Consider the Clifford unitary $U$ transforming stabilizer and $Y$ operators in the following way
\begin{align}
  & U^\dagger~Y_{i,j}~U= g_{i,j}\nonumber \\
  & U^\dagger~g_{i,j}~U= Y_{i+1,j+1},\label{eq_duality_map}
\end{align}
where $g_{i,j}$ is defined in the main text, and $i,j \in \mathbb{Z}$. Diagrammatically, it acts as a half-translation in both $i$ and $j$ directions (see Fig. \ref{fig_duality_map}a) such that $U^2$ is just the lattice translation $(i,j)\mapsto (i+1,j+1)$.

It is straightforward to verify that this transformation yields the right commutation relations for the images of $g$ and $Y$ operators. Clearly, the images of $Y$ operators commute among themselves, and similarly for the images of the $g$ operators. On the other hand, note that a stabilizer operator $g$ anti-commutes only with the four single qubit $Y$ operators acting on its corners. The $U$ transformation maps $g$ to a single qubit $Y$ operator and maps the four $Y$ operators to the four neighboring stabilizers, keeping the anti-commutation relations.

To uniquely specify the unitary $U$, one has to define its action on a complete basis of Pauli strings. The set of $Y$ and $g$ operators is not a complete basis for Pauli strings on an infinite plane and as such, the transformation in \eqref{eq_duality_map} does not fully specify $U$. However, since the projective random quantum circuits at $p_z=0$ are only comprised of $g$ and $Y$ measurements, no matter how one extends Eq.\eqref{eq_duality_map} to a complete basis, the Clifford $U$ maps a projective random quantum circuit chosen with probability distribution corresponding to $p_z=0$ and $p_y=p$ to a projective random quantum circuit chosen according to the probability distribution corresponding to $p_z=0$ and $p_y=1-p$.

Moreover, if the stabilizer set describing the state of the system is generated only by operators comprised of $g$ and $Y$ operators, Eq.\eqref{eq_duality_map} is enough to specify the image of the wave function under $U$ transformations. It also ensures that $U$ keeps the local entanglement structure of the state intact, i.e. changing the entanglement of a region by at most a term proportional to the region's area. These considerations then enforce the $p_z=0$ line of the phase diagram to be symmetric around $p_y=0.5$ point.

The second duality maps the $Y$ and $g$ operators to the gauge operators of the Bacon-Shor code. More specifically, consider the Clifford unitary $\tilde{U}$ which transforms the stabilizer and $Y$ operators as
\begin{align}
  & \tilde{U}^\dagger~Y_{i,j}~\tilde{U}= Z_{i,j}Z_{i+1,j}\nonumber \\
  & \tilde{U}^\dagger~g_{i,j}~\tilde{U}= X_{i+1,j}X_{i+1,j+1}.\label{eq_duality_map_bacon}
\end{align}
It is illustrated diagrammatically in Fig. \ref{fig_duality_map}b. As before, we consider the infinite plane geometry.  It is easy to verify that $\tilde{U}$ preserves the commutation relations and hence could be extended to a complete unitary. We note that the measurements of $g$ and $Y$ operators, when viewed in the dual picture, resembles the syndrome measurements of the Bacon-Shor subsystem code in the active error correction scheme.

Closely related to the Bacon-Shore code is the quantum compass model on a square lattice defined via the following Hamiltonian,
\begin{equation}
  H=-J_Z\sum_{i,j} Z_{i,j} Z_{i+1,j} -  J_X\sum_{i,j}X_{i,j}X_{i,j+1}.
\end{equation}
Interestingly, tuning the coupling ratio $J_Z/J_X$ drives the system through a first order phase transition at $J_Z=J_X$\cite{PhysRevB.72.024448,PhysRevB.75.144401}.

Given the close relation between the quantum compass model and the projective random quantum circuit along $p_z=0$ line, and noting the apparent divergence of order parameters at the self-dual critical point in the thermodynamic limit, one might suspect that the transition at $p_y=0.5$ is also a first order phase transition. One way to test this hypothesis, is to look at a closely related quantum random circuit for which, at $(p_z,p_y)=(0,p_y)$, the ratio of the number of single qubit-$Y$ measurements to total measurements is fixed to be exactly $p_y$ rather than having a Gaussian distribution centered on $p_y$ with a $1/\sqrt{N}$ variance (which is the case for the original projective random circuit model.)  If it was a first order transition, one would expect a much sharper transition in the modified circuit model, because the $1/\sqrt{N}$ variance which would round out the transition at finite $N$, has been eliminated. Fig.\ref{fig_firstorder}a shows the TEE as a function of $p_y$ along the $p_z=0$ in the modified circuit with fixed ratio of measurements. The analogous result in the original circuit model is shown Fig.\ref{fig_firstorder}b for comparison. As is evident from the figures, there is not any noticeable difference between the two, which suggests that the phase transition at $p_y=0.5$ is not a first order transition.

\begin{figure}
  \includegraphics[width=\columnwidth]{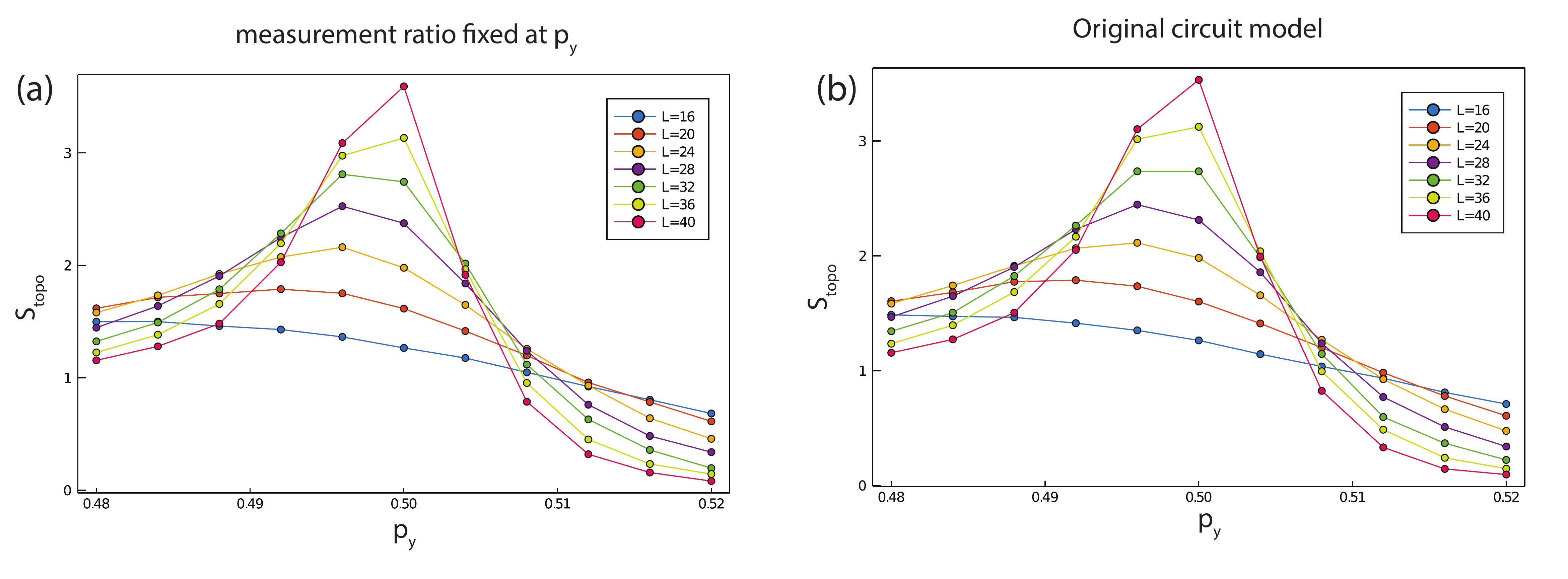}
  \caption{a) TEE as a function of $p_y$ along $p_z=0$ line of the projective random circuit, when the ratio of Pauli-$Y$ measurements is fixed to be $p_y$ at each random circuit realization. b) The analogous result for the original circuit}
  \label{fig_firstorder}
\end{figure}

\subsection{Jacobi-$\theta$ functions and Dedekind-$\eta$ function}\label{apx_theta_eta}

The Jacobi theta functions are defined for two complex numbers $z$ and $\tau$, where $\tau$ is assumed to lie in the upper half plane. For a given $\tau$, the squared nome $q$ is defined as
\begin{equation}
  q(\tau)=e^{2\pi i \tau}.
\end{equation}
In terms of $q$ and $z$, the Jacobi-$\theta$ functions are defined as
\begin{align}
  \theta_1(z|\tau)=&\sum_{n \in \mathbb{Z}}(-1)^{n-1/2}q^{\frac{1}{2}(n+1/2)^2}e^{i(2n+1)z},\\
  \theta_2(z|\tau)=&\sum_{n \in \mathbb{Z}}q^{\frac{1}{2}(n+1/2)^2}e^{i(2n+1)z},\\
  \theta_3(z|\tau)=&\sum_{n \in \mathbb{Z}}q^{\frac{1}{2}n^2}e^{2inz},\\
  \theta_4(z|\tau)=&\sum_{n \in \mathbb{Z}}(-1)^n q^{\frac{1}{2}n^2}e^{2inz},\\
\end{align}
and the Dedekind-$\eta$ function is given as
\begin{align}
  \eta(\tau)=q^{1/24} \prod_{k=1}^\infty (1-q^k).
\end{align}
The one variable $\theta_\nu(\tau)$ function is defined as,
\begin{equation}
  \theta_\nu(\tau)=\theta_\nu(0|\tau), \qquad \nu=1,2,3,4.
\end{equation}
In particular, we have:
\begin{equation}\label{eq_theta3}
  \theta_3(\tau)=\sum_{n \in \mathbb{Z}}q^{\frac{1}{2}n^2}.
\end{equation}
To understand the the asymptotic form of the entanglement function $S^\text{QLM}(x)$ for $x\ll L$, we need to expand the function $J(u)$,
\begin{equation}
  J_\lambda(u)=\log\qty(\frac{\theta_3(i \lambda u) \theta_3(i \lambda (1-u))}{\eta(2iu)\eta(2i(1-u))})
\end{equation}
for small $u$. To that end, we use the following identities, related to the modular transformation properties of the Jacobi $\theta$ functions
\begin{align}
  \theta_3(\tau)=&(-i \tau)^{-1/2} \theta_3(-1/\tau)\\
  \eta(\tau)=&(-i\tau)^{-1/2} \eta(-1/\tau).
\end{align}
Using the above identities, it is easy to see that for $\lambda u \ll 1$,
\begin{align}
  \theta_3(i\lambda u)\approx &(\lambda u)^{-1/2}\\
  \eta(2iu)\approx &(-2u)^{-1/2} e^{-\frac{\pi}{24 u}},
\end{align}
which results in the $1/u$ leading term in the expansion of $J_\lambda(u)$ at $\lambda u\ll 1$,
\begin{equation}
  J_\lambda(u)\approx \frac{\pi}{24}u^{-1}+\text{const.}
\end{equation}
Note that the coefficient of the leading term is independent of $\lambda$.

The Nemo package\cite{nemo} was used for numerical evaluation of $\theta_3$ and $\eta$ functions in this study.

\section{Estimating Errors}

We briefly review the procedure which was utilized in this study to find critical probabilities and critical exponents and explain how their errors were estimated. In general ancilla order parameter yields smoother diagrams and needs less averages, and hence we mostly use $S_a$ to extract exponents.

Consider a quantity $S$ which has the following expected scaling form
\begin{equation}\label{eq_generic_scaling_form}
  S(p,L)=F((p-p_c)L^{1/\nu}),
\end{equation}
with $F(x)$ an unknown function. This is the expected scaling form for $S_\text{topo}(p,L)$, assuming $\gamma=0$, as well as $S_a$, when measured at $t\propto L^z$. To find the value of $p_c$ and $\nu$ that result in the best data collapse, we use the objective function $\epsilon(p_c,\nu)$ which is defined as
\begin{equation}\label{eq_objfunction}
  \epsilon(p_c,\nu)=\sum_{i=2}^{n-1}(y_i-\bar y_i)^2,
\end{equation}
where
\begin{equation}\label{eq_ybar}
  \bar y_i=y_{i-1}+\frac{y_{i+1}-y_{i-1}}{x_{i+1}-x_{i-1}}(x_i-x_{i-1}),
\end{equation}
with $y_i=S(p_i,L_i)$, $x_i=(p_i-p_c)L_i^{1/\nu}$ and $n$ is the the total number of data points. The index $i$ enumerates over all data points sorted such that $x_1<x_2<\cdots<x_n$. $\bar y_i$, as defined in Eq.\eqref{eq_ybar}, is the expected value of $y_i$ when it lies on the line passing through $(x_{i-1},y_{i-1})$ and $(x_{i+1},y_{i+1})$. Hence, for a perfect collapse and in the limit of infinite data points, one gets $\epsilon(p_c,\nu)=0$. Accordingly, to find the best collapse, we find  $p_c^\ast$ and $\nu^\ast$
values  which minimize the objective function $\epsilon$. To estimate the error in determining $p_c^\ast$ and $\nu^\ast$, we find the interval around $(p_c^\ast,\nu^\ast)$, outside of which $\epsilon(p_c,\nu) > 2\, \epsilon(p_c^\ast,\nu^\ast)$ (see Fig.\ref{fig_heatmap}).

\begin{figure}
  \includegraphics[width=\columnwidth]{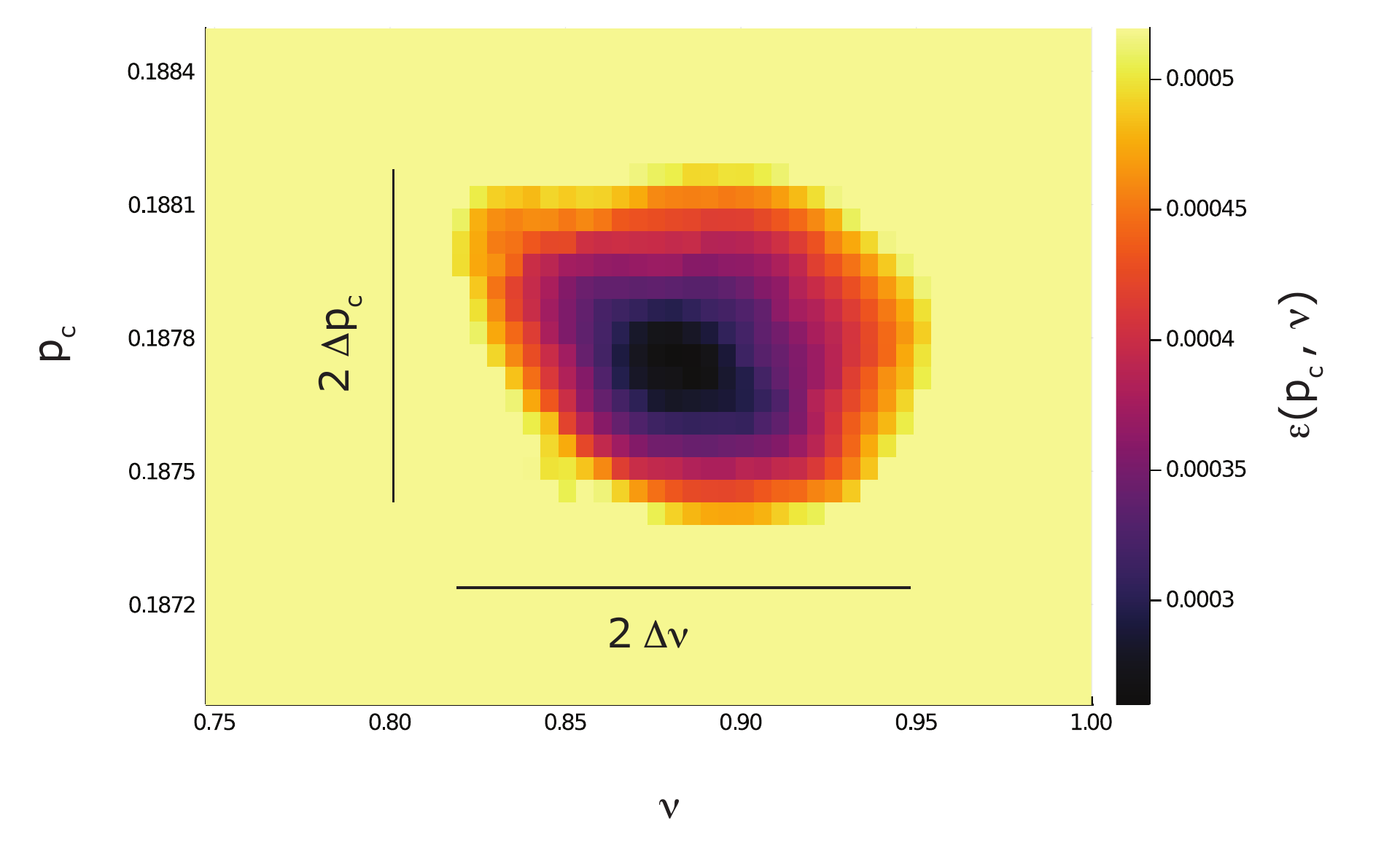}
  \caption{The objective function as a function of $p_c$ and $\nu$ associated with the ancilla entropy data points near the percolation critical point. Black denotes the minimum value $\epsilon(p_c^\ast,\nu^\ast)$ and the light yellow correspond to values greater or equal to $ 2\, \epsilon(p_c^\ast,\nu^\ast)$. The corresponding error intervals are specified on the plot.}
  \label{fig_heatmap}
\end{figure}

In case the scaling form has an extra scaling pre-factor, i.e.
\begin{equation}
  S(p,L)=L^\gamma \,F((p-p_c)L^{1/\nu}),
\end{equation}
we use basically the same procedure, but with $y_i$ defined as $y_i=S(p,L)/L^\gamma$ and using a slightly different objective function, defined as
\begin{equation}
  \tilde \epsilon(p_c,\nu,\gamma)=\sum_{i=2}^{n-1}(\log(y_i)-\log(\bar y_i))^2.
\end{equation}
We use $\log(y_i)$ instead of $y_i$ to make the objective function invariant under an overall rescaling of $y_i$. Otherwise, one can make the objective function in Eq.\eqref{eq_objfunction} arbitrarily small by choosing large enough $\gamma$.

Finally, the error bars for the best fit values of $a$, $b$ and $\lambda$ parameters correspond to the standard errors found via least square fitting which reflects the quality of the fit.

\clearpage
\widetext
\section{Supplementary Figures}\label{apx_supfig}
\begin{figure}[!h]
  \includegraphics[width=\textwidth]{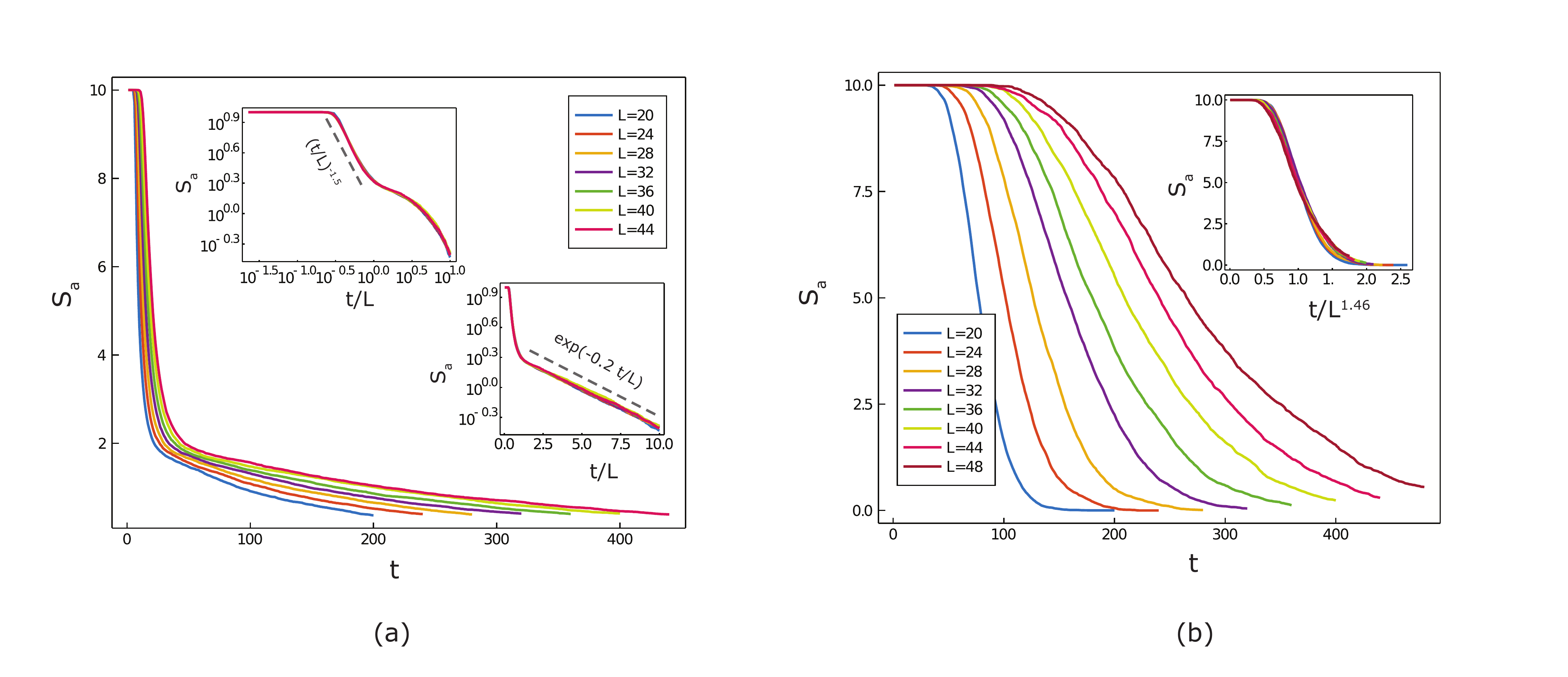}
  \caption{(a) The ancilla entropy versus time at the percolation critical point, for different system sizes. The insets are the same quantity as a function of $t/L$ in log-log and semi-log scales, showing that $S_a$ is a function of $(t/L)$ at the critical point in agreement with $z=1$. Similar to $(1+1)$D circuits, the ancilla order parameter initially decays as a power law until $t=O(L)$, and then falls off exponentially. (b) the ancilla entropy versus time at the self-dual critical point $(p_z,p_y)=(0,0.5)$ which shows a quite different time-dependence at the self-dual critical point compared to the percolation critical point. By plotting $S_a$ for different system sizes as a function of $t/L^z$ for various $z$, we find $z=1.46$ results in the best collapse (shown in the inset), suggesting that the dynamical critical exponent at the self-dual critical point is $z=1.46(8)$.  }
  \label{fig_z}
\end{figure}

\begin{figure}[h]
  \centering \includegraphics[width=0.5\textwidth]{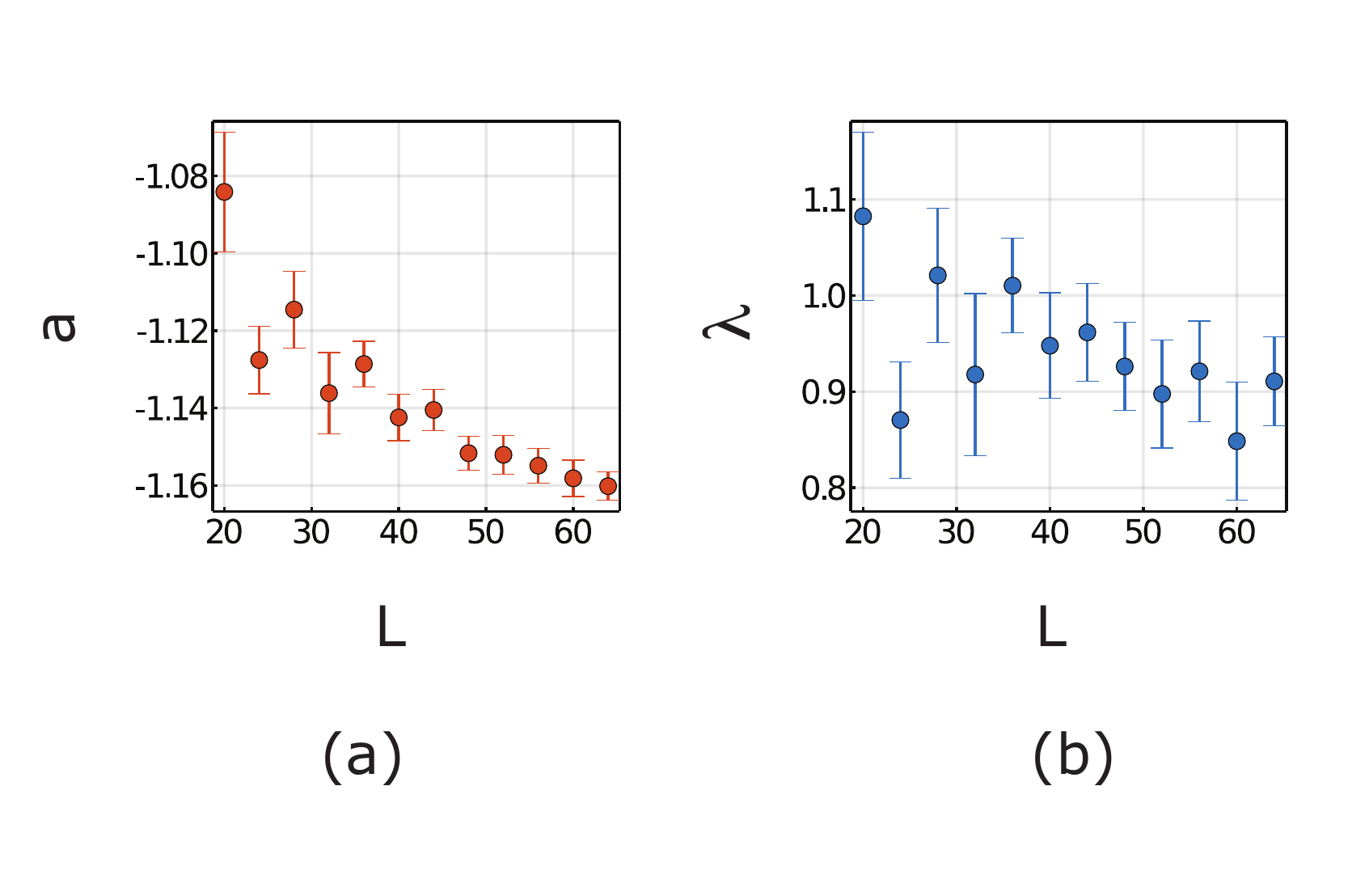}
  \caption{The best fit value of the (a) $a$ and (b) $\lambda$ parameters of the $S^\text{qlm}(x)$ scaling function at the percolation crtitical point (on the $p_z$ axis for the projective and hybrid random circuits described in the main text (i.e. $p_y  = 0$ or $p_u = 0$, respectively). For the sub-leading term in $S^\text{qlm}(x)$ to be universal, $a$ and $\lambda$ have to be independent of the system size. Although we observe a weak dependence of the $a$ and $\lambda$ on $L$, it could be either due to inherent error in pinpointing the exact location of the critical point or the finite size effects, in which case $a$ and $\lambda$ should plateau at large $L$s. Moreover note that $\lambda$ parameter has rather large error bars, which is related to the nonlinear dependence of the $S^\text{qlm}$ form on $\lambda$. }
  \label{fig_aandlmab}
\end{figure}

\begin{figure}
\includegraphics[width=\textwidth]{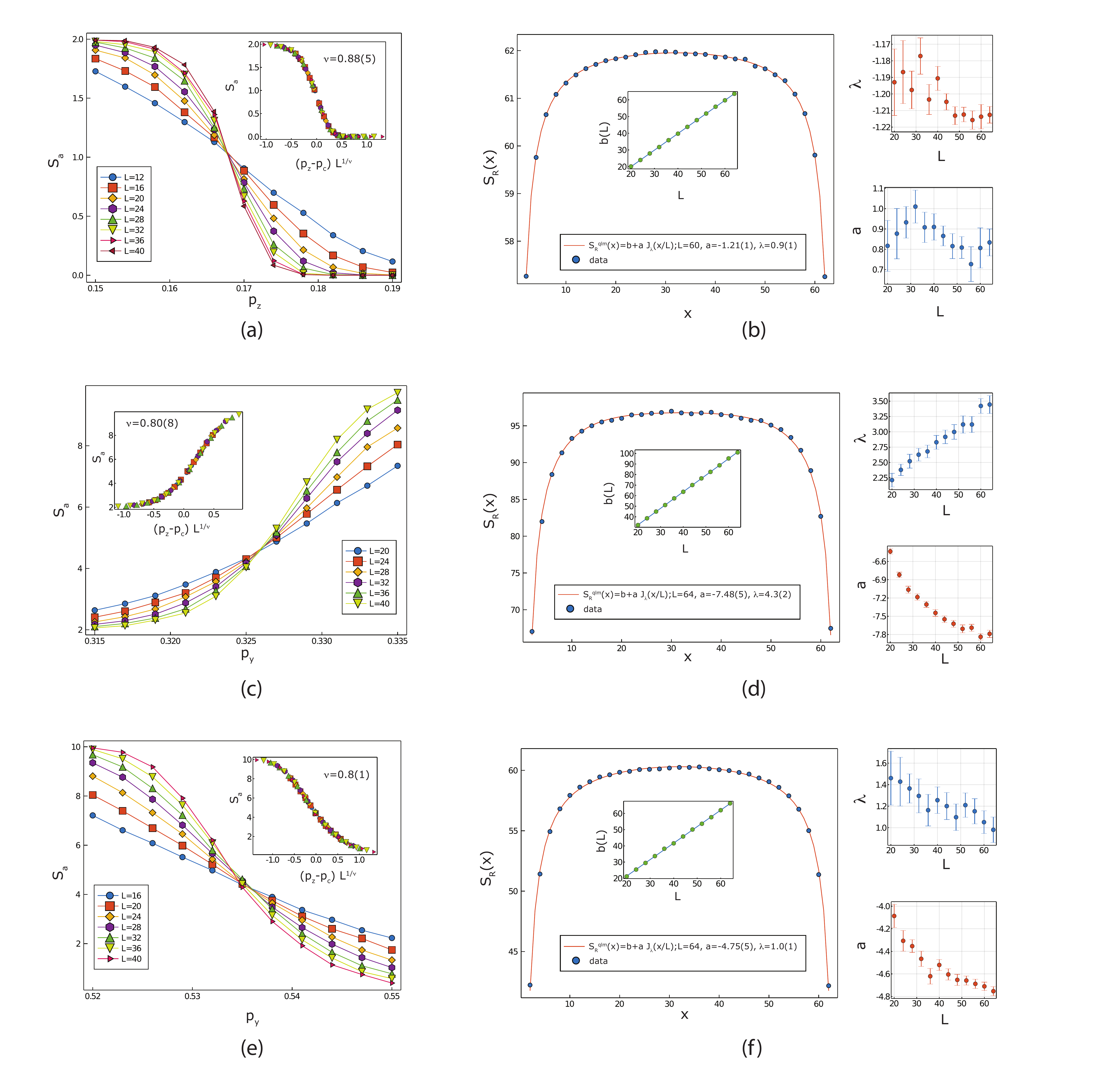}
\caption{Additional plots for the projective random circuit, showing a typical transition on the boundary of topological phase and trivial phase (first row), topological phase and volume law phase (second row) and volume law phase and trivial phase (third row).  (a) the ancilla order parameter measured at $t=5L$ versus $p_z$ for fixed $p_y=0.1$ near the topological to trivial phase transition with the corresponding collapse with $p_c=0.168(2)$ and $\nu=0.88(5)$ shown in the inset. (b) The entanglement of the cylindrical region, $S_R(x)$, as a function of x at $(p_z,p_y)=(0.168,0.1)$ for $L=64$, with best fit of scaling function $S^\text{qlm}$. The inset shows the $L$ dependence of the $b$ parameter in $S^\text{qlm}$ scaling function. On the right, is the $L$ dependence of the best fit values of $a$ and $\lambda$
 parameters versus $L$. (c) $S_a$ measured at $t=L$ versus $p_y$ at fixed $p_z=0.1$ showing the phase transition from topological phase to volume law. The inset shows the corresponding data collapse with $p_c=0.326(1)$ and $\nu=0.80(8)$. (d)  $S_R(x)$ as a function of x at $(p_z,p_y)=(0.1,0.326)$ with the best fit of $S^\text{qlm}$ scaling function, with the inset, top right and bottom right panels showing the $L$ dependences of best fit values of the parameters $b$, $\lambda$ and $a$ respectively. (e)$S_a$ measured at $t=3/2  \,L$
  versus $p_y$ at fixed $p_z=0.1$ showing the phase transition from volume law phase to trivial. The inset shows the corresponding data collapse with $p_c=0.535(1)$ and $\nu=0.8(1)$. (f)$S_R$ at $(p_z,p_y)=(0.1,0.535)$ with best fit of scaling function $S^\text{qlm}$, the best fit value of $b$ as a function of $L$ (inset) and the best fit values of $a$ and $\lambda$ versus $L$ on the right.}
\label{fig_yzmiddle}
\end{figure}

\begin{figure}
\includegraphics[width=\textwidth]{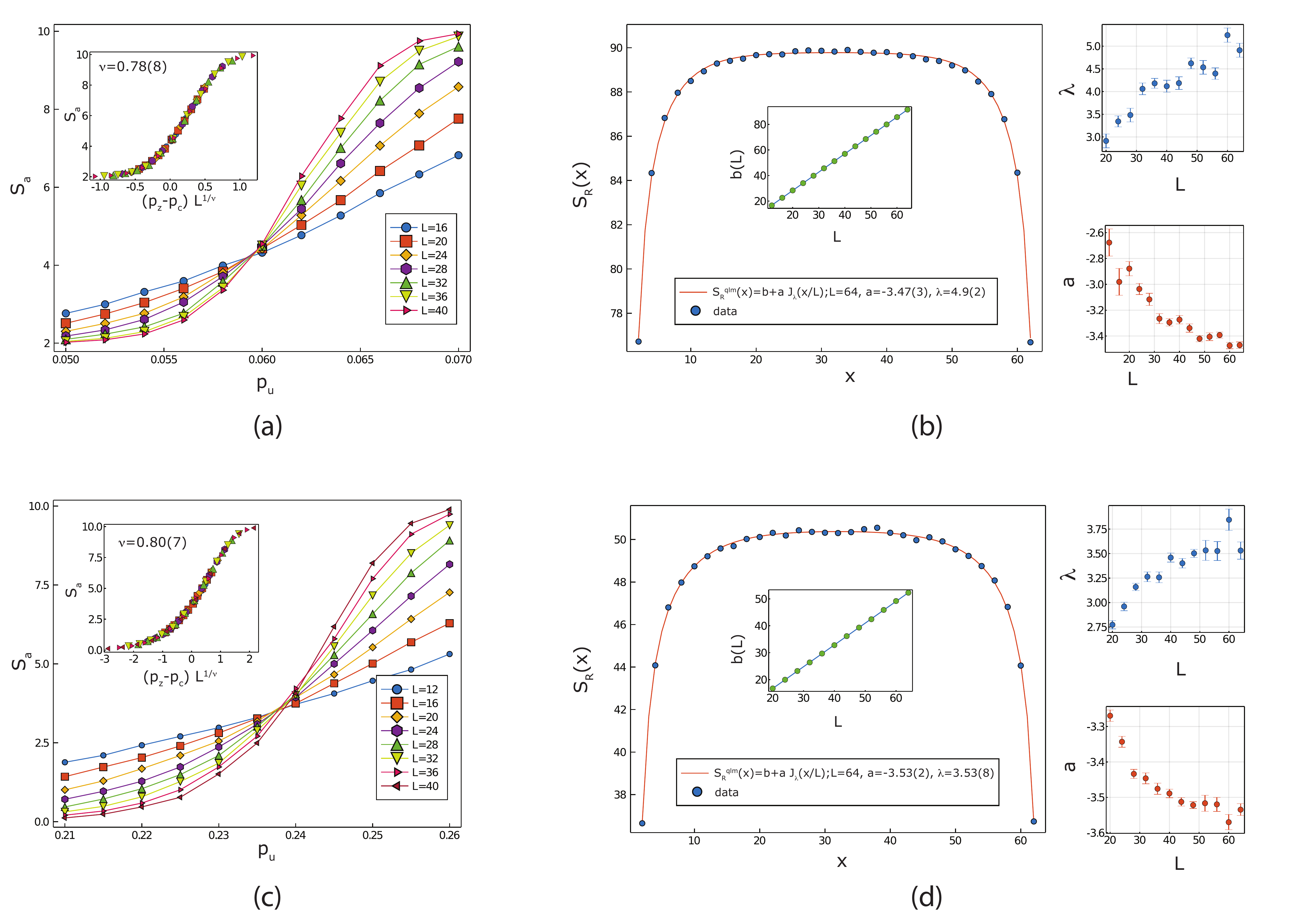}
\caption{Additional plots for the hybrid random circuit corresponding to the transitions from the topological phase to the volume law phase along the $p_z=0$ line (first row) and the phase transition from  trivial phase to volume law phase along the hypotenuse $p_z+p_u=1$ (second row).  (a) the ancilla order parameter measured at $t=3/4\, L$ versus $p_u$ along the $p_z=0$ axis near the topological to volume law phase transition with the corresponding collapse with $p_c=0.059(1)$ and $\nu=0.78(8)$ shown in the inset. (b) The entanglement of the cylindrical region, $S_R(x)$, as a function of x at $(p_z,p_u)=(0,0.059)$ for $L=64$, with best fit of scaling function $S^\text{qlm}$. The inset shows the $L$ dependence of the $b$ parameter in $S^\text{qlm}$
scaling function. The $L$ dependence of the best fit values of $a$ and $\lambda$
 parameters are shown on the right. (c) $S_a$ measured at $t=L$ versus $p_u$ along the $p_z+p_u=1$ line, near the phase transition from trivial phase to volume law phase. The corresponding data collapse with $p_c=0.238(2)$ and $\nu=0.80(7)$ is shown in the inset. (d)  $S_R(x)$ as a function of x at $(p_z,p_u)=(0.762,0.238)$ with the best fit of $S^\text{qlm}$ scaling function. The inset shows the $L$ dependence of the $b$ parameter in $S^\text{qlm}$ scaling function, while the dependence of the best fit values of $a$ and $\lambda$ parameters is shown on the right. }
\label{fig_uz}
\end{figure}
